\documentclass[11pt,a4paper,pdftex]{article}
\usepackage{amsmath, amscd, amssymb, amsthm, latexsym,
}
\usepackage[colorlinks=true,linkcolor=red]{hyperref}
\usepackage{stackrel}
\usepackage{pdflscape}

\newcommand{\M}{\mathbb{M}}

\newcommand{\nl}{\medskip\noindent}

\newtheorem{theorem}{Theorem}

\newtheorem{lemma}[theorem]{Lemma}

\theoremstyle{definition}
\newtheorem{definition}[theorem]{Definition}

\theoremstyle{remark}

\newtheorem{remark}[theorem]{Remark}
\newtheorem{example}[theorem]{Example}

\newcommand \str [1] {$\langle  {#1} \rangle$}

\newlength{\transwidth}

\begin{document}
\title{ Decidability of the existential fragment of \\
some infinitely generated trace monoids:\\
an application to ordinals}

\author{Alexis B\`{e}s\\
	Universit\'e Paris-Est, LACL (EA 4219), UPEC, Cr\'eteil, France \\
	\and Christian Choffrut\\
	IRIF, CNRS and Universit\'e Paris 7 Denis Diderot, 
	France} 

\maketitle

\begin{abstract}
	Diekert, Matiyasevich and Muscholl proved that the existential first-order theory of a trace monoid over a finite alphabet is decidable. We extend this result to a natural class of trace monoids with infinitely many generators. As an application, we prove that for every ordinal $\lambda$ less than $\varepsilon_0$, the existential theory of the set of successor ordinals less than $\lambda$ equipped with multiplication is decidable.
	
\end{abstract}

\maketitle

%
%
%

\section{Introduction}

Since the publication of the surprising result of Makanin
showing that it is  decidable whether or not an equation in a free monoid with constants
has a solution, 
research continued in different directions.
Notably,  Schulz showed that the decidability 
still holds
when each variable is bound to be interpreted in a predefined regular subset
of the free monoid, i.e., 
as subset recognized by some finite automaton, \cite{schulz90}. 
This allowed 
to extend Makanin's result to trace monoids which can be viewed as free monoids 
where some pairs of generators may commute. Indeed, building on top of 
the results of Makanin and Schulz, Diekert et al. 
were able to prove that it can be decided whether or not
a system of equations in a trace monoid has a solution and more generally
whether or not a sentence of the existential fragment  the theory
of trace monoids provided with the concatenation is valid.

\medskip Our contribution consists of going one step further by considering 
some trace monoids over a countable alphabet, namely those that are the inverse image 
in a generator-to-generator substitution into a finitely generated trace monoid
and to show that the above result still holds in this framework. The idea can be summarized as follows.
Considering infinite generators is no problem as far as equations are concerned 
because  a solution, if it  exists,   can always be assumed 
to map into the submonoid  of the generators 
 appearing in the constants. In contrast, the existential fragment requires  
 the possibility of expressing the negation of an equality. When the monoid is 
 finitely generated this can be done by a finite disjunction of equations
 containing new constants. For infinitely generated trace monoids, this leads to an infinite
 disjunction with infinite constants. The idea is then to observe that the actual 
 values of these constants are irrelevant and can be reduced to a 
number of values that can be bounded a priori, allowing thus
to resort to the case of finitely generated trace monoids.

We apply our result to an issue concerning the 
multiplicative structure of ordinals. Let us recall that the first-order theory of an 
ordinal $\alpha$ with the multiplication as unique operation is undecidable if and only if 
this ordinal is greater  than or equal to $\omega^{\omega}$, see \cite{bes00}. 
We are not aware of any investigation on which fragment, if any,  
of such an ordinal is decidable, except a paper of the present authors,
\cite{besch}.  It happens that the successor 
ordinals less than a multiplicatively closed ordinal $\lambda$ 
($\alpha, \beta<\lambda$ implies  $\alpha\times  \beta<\lambda$) 
form a monoid $\+S_{\lambda}$ which is a free product
of an infinitely generated free commutative monoid  (more precisely
the monoid generated by the ordinary prime integers) and 
an infinitely generated free monoid. The result on traces implies that 
the existential fragment of $\+S_{\lambda}$ is decidable. More precisely we prove that if $\lambda$ is less than $\varepsilon_0$ then the existential theory of the structure $\langle \+S_{\lambda}:  \times,   \{\alpha\}_{\alpha \in \+S_{\lambda}}   \} \rangle$
is decidable. Recall that $\varepsilon_0$ is the least fixed point of the function $x \mapsto \omega^x$. The condition $\lambda<\varepsilon_0$ is here to ensure  that  one can perform effectively operations on constants given by their Cantor normal form.

 This does not settle the 
problem of the decidability of the multiplicative structure of $\lambda$ but we hope
it arises the curiosity of some researchers.

\section{Trace monoids}

Our purpose in this section is to show how the result of 
Diekert and al. \cite{diekert99} for finitely 
generated trace monoid  extends to some type of infinitely generated trace 
monoids. We first give a definition of trace monoids with possibly infinite generators.

\subsection{Traces over possibly infinite alphabets}

The notion of finitely generated trace monoids can be  recovered from the next
definition. 
\begin{definition}
\label{de:infinite}
Let $\Sigma$ be a countable set, $n>0$ an integer, $\Sigma=\bigcup^{n}_{i=1} \Sigma_{i}$
a decomposition  and $I\subseteq \{1, \ldots, n\}\times \{1, \ldots, n\}$ a
 symmetric relation. The relation 
$$
\{(ab,ba) \mid a\in \Sigma_{i}, b\in   \Sigma_{j}, \  (i,j)\in I\}
$$
is the  \emph{independence relation} on
$\Sigma$ induced  by $I$. We denote by   $\equiv_{I}$
the congruence on the free monoid $\Sigma^*$ generated 
by the independence relation
and by $\M(\Sigma, I)$ the quotient monoid $\Sigma^*/\equiv_{I}$, 
also known as the \emph{trace monoid}. The \emph{canonical morphism}
is the mapping which associates with an element of $\Sigma^*$ its
class modulo $\equiv_{I}$. When $\Sigma$ and $I$ are  clear from the context we simply write $\M$.
\end{definition}

Observe that in contrast to the finite case, we do not impose the relation 
$I$ to be irreflexive. If it  were, we would necessarily have a free (noncommutative) 
submonoid.

\begin{example} The free product of a  free monoid 
  generated by  $\Sigma_{1}$ and a 
 free commutative monoid  generated by  a disjoint subset $\Sigma_{2}$ is defined
 by the conditions $\Sigma= \Sigma_{1}\cup \Sigma_{2}$,  $n=2$ 
 and $I=\{(2,2)\}$.
The free product 
of three free commutative monoids is defined by taking three disjoint subsets
$\Sigma_{1}, \Sigma_{2}, \Sigma_{3}$ and considering
the independence relation $I=\{(i,i)\mid i=1,2,3\}$.
The direct product 
of two free  monoids is defined by taking two disjoint subsets
$\Sigma_{1}, \Sigma_{2}$ and considering
the independence relation $I=\{(1,2), (2,1)\}$.
 \end{example}

We will need no sophisticated result on trace monoids, 
only the following  simple combinatorial result stating under which conditions
two traces are different,  see \cite[page 8]{diekert99}. Further reading
on this theory is referred,   for example,  to \cite{diekert90, diekert95}.
\begin{lemma}
\label{le:different-traces}
Two elements 
$u,v\in \M(\Sigma, I)$ are different
exactly under one of the following conditions  (up to exchanging 
the roles of $u$ and $v$)

\begin{enumerate}

\item \label{it:a1} 
$u$ is a strict prefix of $v$.

\item \label{it:a2} 
there exist  $w, w_{1}, w_{2}, w_{3} \in \M(\Sigma, I)$, $a\in \Sigma$ such that
$u=w a w_{1}, v=w w_{2}a w_{3}$ and $w_{2}$ contains no occurrence of $a$ but contains
an occurrence of some $b\not=a$ which does not commute with $a$, i.e., $(a,b)\not\in I$.  

\item \label{it:a3} 
there exist  $w, w_{1}, w_{2},  \in \M(\Sigma, I)$, $a\in \Sigma$ such that
$u=w a w_{1}, v=w w_{2}$ and $w_{2}$  contains no $a$. 

\end{enumerate}
\end{lemma}

\begin{proof}
Clearly, the conditions \ref{it:a1}, \ref{it:a2}  and \ref{it:a3} are sufficient. We prove that they are necessary.
Let $w$ be the longest common prefix of $u$ and $v$
(such a prefix is guaranteed by  Levi's Lemma, see \cite[Prop. 1.3]{cp85}). If $w=u$ this is condition \ref{it:a1}. If 
$u=waw_{1}$  and $v= ww_{2} $ where $w_{2}$ contains no $a$ this is item \ref{it:a3}.
Otherwise set $v=w  w_{2}a w_{3}$
where $w_{2}$ contains no occurrence of $a$. Then $w_{2}$ certainly contains some $b$
which does not commute with $a$.
\end{proof}

\subsection{Trace monoids as logical structures}

The logical structures on  trace monoid considered in this work contain at least the equality as nonlogical symbol, 
 the product of traces as function and all their elements as constants.

We consider the structure $\+M=$\str{\M(\Sigma,I);=,\cdot, \{u\}_{u\in \M(\Sigma,I)}}.
We are given a countable set of elements called \emph{variables}. The family of 
\emph{terms} is defined inductively by the conditions: a variable is a term, a constant
is a term and if $t$ and $t'$ are two terms, so is $t\cdot t'$. If needed, we may write
a term as $t(x_{1}, \ldots, x_{n})$ to signify that the variables occurring in $t$
are among $x_{1}, \ldots, x_{n}$. A \emph{basic predicate} is of the form $t=t'$
where $t$ and $t'$ are two terms.

As much as possible, we use the lower case $a,b, \ldots$ for elements in 
$\Sigma$ which we do not distinguish from generators of $ \M(\Sigma, I)$, 
and lower case letters $u,v, \dots$ for traces.

\subsection{The case of finitely generated trace monoids}

  We recall the result of Diekert et al. \cite{diekert99}.
In this case, the generator set   $\Sigma$ 
is finite and its decomposition consists of 
the union of all singletons $\{a\}$ for $a\in \Sigma$.
The  structure has additional predicates whose definition requires the following
notion.

\begin{definition}
\label{de:regular-constraints}
A subset
$K$ of $\M(\Sigma, I)$ is \emph{regular} 
if it is the
 image, in the canonical mapping of $\Sigma^*$ onto
$\M(\Sigma, I)$,  of  a  regular 
subset $K_{0}\subseteq \Sigma^*$ which is $I$-\emph{closed}, in the sense that the condition
($u \in K_{0}$ and $u\equiv_{I} v$) implies $v\in K_{0}$. The family of 
regular subsets is denoted by $\+K$. 
\end{definition}

Observe that each singleton of $\M(\Sigma, I)$ is regular 
since it is the canonical image of  a finite subset of $\Sigma^*$.

\begin{theorem}[\cite{diekert99}]
\label{th:finite-case}
Given a finitely generated trace monoid, the existential fragment of the first-order theory of the
structure $\langle \M(\Sigma, I); =, \cdot, \{K\}_{K\in \+K }\rangle $, where each $K$ is viewed as a unary predicate, is decidable.
\end{theorem}

\subsection{The case of infinitely generated trace monoids}

The objective is to prove the following.

\begin{theorem}
\label{th:infinite-case}
Given a trace monoid as in Definition \ref{de:infinite},
the existential first-order theory of the structure $\langle \M(\Sigma,I);=,\cdot,  \{u\}_{u\in \M(\Sigma,I)} \rangle$
is decidable.
\end{theorem}

\begin{proof}

We fix the decomposition  $\Sigma=\bigcup^{n}_{i=1} \Sigma_{i}$ and 
the independence relation $I$ as in Definition \ref{de:infinite}
and we write $\M$ for $\M(\Sigma,I)$ whenever
no confusion may arise.
A formula of the existential fragment of the 
theory of the trace monoid is equivalent to a formula of the form 
%
\begin{equation}
\label{eq:system-E}
\begin{array}{l}
\exists x_{1}, \ldots, x_{n}\ \displaystyle \bigvee^{p}_{k=1}   (\+E^{+}_{k} \wedge \+E^{-}_{k})\\
 \text{ where each } \+E^{+}_{k} \text{ is of the form}\\
\displaystyle \bigwedge_{i\in I} t_{i}(x_{1}, \ldots, x_{n}) =   s_{i} (x_{1}, \ldots, x_{n})\\
 \text{and  where each } \+E^{-}_{k} \text{ is of the form}\\
\displaystyle \bigwedge_{j\in J} t'_{j}(x_{1}, \ldots, x_{n})  \not=  s'_{j}(x_{1}, \ldots, x_{n}) \\
\end{array}
\end{equation}

Such a formula is \emph{satisfiable} if there exists  an interpretation 
$$\theta: \{x_{1}, \ldots, x_{n}\} \rightarrow  \M(\Sigma,I),
$$
such that  $\theta(t_{i})=   \theta(s_{i})$ for   $i\in I$ and  $ \theta(t'_{j})\not=    \theta(s'_{j})$
for   $j\in J$. 

\medskip For simplification purposes we replace the basic predicates by simpler predicates
by applying the following rules. 

\begin{itemize}
	\item
 Constants: a constant $u\not=1$ in a term is replaced 
by a new variable $x$ and by adding the new condition $x=u$. All occurrences of the empty
trace $1$ are deleted except if the left- or right-hand side is reduced to $1$.

\item Equation $y_{1}\cdots y_{n}=y_{n+1}\cdots y_{m}$ is replaced by the 
condition $\bigwedge^{n}_{i=1}y_{i}=1$ if the right-hand side is reduced to the empty 
trace. Similarly if the left-hand side is reduced to $1$. Otherwise
the equation is of the form
$y_{1}\cdots y_{n}= y_{n+1}\cdots y_{m}$ with $1< n<m$. 
Then we  introduce $m$ new variables $z_{i}$ and replace the equation by  a  conjunction
of simple predicates
\begin{equation}
\label{eq:product-by-increment}
\begin{array}{l}
z_{1}=y_{1} \wedge z_{2}=z_{1} y_{2} \wedge \cdots \wedge z_{n}=z_{n-1} y_{n}\\
\wedge\  z_{n+1}=y_{n+1} \wedge z_{n+2}=z_{n+1} y_{n+2} \wedge \cdots \wedge z_{m}=z_{m-1} y_{m}\\
\wedge\  z_{n}= z_{m}
\end{array}
\end{equation}

\item Inequation $y_{1}\cdots y_{n}\not=y_{n+1}\cdots y_{m}$ is replaced by the 
condition $\bigvee^{n}_{i=1}y_{i} \ne 1$ if the right-hand side is reduced to the empty 
trace. Similarly if the left-hand side is reduced to $1$. Otherwise
the inequation is of the form 
$y_{1}\cdots y_{n}\not= y_{n+1}\cdots y_{m}$ with $1< n<m$. 
As in the case of equality, we  introduce $m$ new variables $z_{i}$ and replace the equation by  a  conjunction
of simple predicates
\begin{equation}
\label{eq:product-by-increment}
\begin{array}{l}
z_{1}=y_{1} \wedge z_{2}=z_{1} y_{2} \wedge \cdots \wedge z_{n}=z_{n-1} y_{n}\\
\wedge\  z_{n+1}=y_{n+1} \wedge z_{n+2}=z_{n+1} y_{n+2} \wedge \cdots \wedge z_{m}=z_{m-1} y_{m}\\
\wedge\  z_{n}\not= z_{m}
\end{array}
\end{equation}

\end{itemize}

Equalities between two variables can be deleted by keeping 
one of the two variables only. By applying the previous rules and the routine Boolean rules, we may 
rewrite expression \ref{eq:system-E} in such a way that all  its basic predicates  are of the form
$x=u$ or $xy=z$ where $x,y,z$ are variables and $u$ is a constant, possibly
equal to $1$, formally

\begin{equation}
\label{eq:positive-clauses}
\+E_{k}^{+} \equiv \bigwedge_{i\in J^{(k)}_{1}}  x_{i}=u_{i} \wedge   \bigwedge_{(i,j,h)\in J^{(k)}_{2}} x_{i}y_{j}=z_{h}
\end{equation}
and
\begin{equation}
\label{eq:negative-clauses}
\+E_{k}^{-} \equiv \bigwedge_{i\in J^{(k)}_{3}}  x_{i}\not=1 \wedge   \bigwedge_{(i,j)\in J^{(k)}_{4}} x_{i}\not=y_{j}
\end{equation}

 Consequently, we may assume that we start off with a disjunction 
 of  \ref{eq:positive-clauses} and  \ref{eq:negative-clauses}, i.e., 
 of the form $\Phi\equiv \exists x_{1}, \ldots, x_{n}\  (\+E^{+} \wedge \+E^{-})$ by omitting  the index $k$.

Let $\Delta\subseteq \Sigma $ be the 
set of generators appearing as a factor of some constant in $\Phi$, i.e., 
 the smallest subset $\Sigma'\subseteq \Sigma$ satisfying the inclusion
$$
\begin{array}{l}
\{u_{i}\mid i\in  J_{1}\}  \subseteq \M(\Sigma',I\cap (\Sigma'\times  \Sigma')).
\end{array}
$$
We claim that there exists a finite collection of finite
subsets $\Delta \subseteq \Gamma_{i}\subseteq \Sigma$, $i=1 \ldots,  \ell$,
and existential formulas $\Phi_{i}$, $i=1 \ldots,  \ell$ 
in $\langle \M_{i}, \cdot,  \+K_{i} \rangle$
where $\M_{i}=\M(\Gamma_{i},  I \cap (\Gamma_{i}\times \Gamma_{i}))$ and
%
$ \+K_{i}=\{L\subseteq \M_{i} \mid L\in \+K\}$,
%
%
such that the following property holds for all assignments $\theta: \{x_{1}, \ldots, x_{n}\} \rightarrow \M(\Sigma,I) $
\begin{equation}
\label{eq:reduction-to-finitely-many-expressions}
\begin{array}{c}
\langle \M(\Sigma,I); \cdot,  \{u\}_{u\in \M(\Sigma,I)}\rangle  \models \Phi(\theta(x_{1}),\ldots, \theta(x_{n}))    \\
\text{ iff there exists some } i\in \{1, \ldots, \ell\} \text{ and some interpretation } \\
\theta_{i}: \{x_{1}, \ldots, x_{n}\}  \rightarrow  \M_{i} \text{ such that }\\
\langle \M_{i}; \cdot,  \+K_{i}\rangle  \models 
\Phi_{i}(\theta_{i}(x_{1}),\ldots, \theta_{i}(x_{n}))  
\end{array}
\end{equation}
%
%
The idea is to restate the different conditions of inequality for traces.
We fix an assignment $\theta$.
For each inequality $x\not=1$ in \ref{eq:negative-clauses},
the condition  $\theta(x)\not=1$ 
is satisfied if and only if  there exist a generator $a\in \Sigma$
and an element $v\in \M$ such that 

\begin{equation}
\label{eq:x-not-1}
\theta(x)=a v
\end{equation} 

Define

\begin{itemize}

\item $\Delta^{(1)}_{\theta}$ is the subset consisting of all $a\in \Sigma$
appearing in \ref{eq:x-not-1} when $x\not=1$ ranges over all inequalities of this form in
\ref{eq:negative-clauses}

\item    $U^{(1)}_{\theta}$ is the subset consisiting of all $v\in \M$
appearing in \ref{eq:x-not-1} when $x\not=1$ ranges over all inequalities of this form in
\ref{eq:negative-clauses}

\end{itemize}

Similarly, for each inequality of the form $x\not=y$ in \ref{eq:negative-clauses},
by Lemma \ref{le:different-traces} 
the condition  $\theta(x)\not=\theta(y)$ is equivalent to the disjunction
of the following three cases.

\nl Case 1: there exist an element $v\in \M$
and a generator $a\in \Sigma$ such that 

\begin{equation}
\label{eq:x-not-y1}
 \theta(y)=\theta(x) a v    
\end{equation}

\nl  Case 2: there exist two integers $(i,j)\not\in I$,
two generators $a,b\in \Sigma$
and five elements $u,v,w,r, t\in \M$ such that 
\begin{equation}
\begin{array}{l}
\label{eq:x-not-y2}
\theta(x)=ua v,  \theta(y)=u  r b t  a w,   a\in \Sigma_{i},  b\in \Sigma_{j}  \text{ and } r,t\\
 \text{  contain no occurrence of } a  
\end{array}
\end{equation}
\nl Case 3:  there exist $a\in \Sigma$,  $u,v,w \in \M$
such that $w$ contains no occurrence of $a$ and 
\begin{equation}
\label{eq:x-not-y3}
  \theta(x)=u av \text{ and }  \theta(y) = uw 
\end{equation}

Define

\begin{itemize}

\item $\Delta^{(2)}_{\theta}$ is the subset consisting of all $a$ and $b$ in $\Sigma$
appearing in \ref{eq:x-not-y1}, \ref{eq:x-not-y2} and  \ref{eq:x-not-y3}, when 
$x\not=y$ ranges over all inequalities of this form in
\ref{eq:negative-clauses}

\item    $U^{(2)}_{\theta}$ is the subset consisting of all $r,t,u, v, w\in \M$
appearing in \ref{eq:x-not-y1}, \ref{eq:x-not-y2} and  \ref{eq:x-not-y3}
 when $x\not=y$ ranges over all inequalities of this form in
\ref{eq:negative-clauses}

\end{itemize}

Set 
$$
\Delta_{\theta}= \Delta^{(1)}_{\theta}\cup \Delta^{(2)}_{\theta}\subseteq \Sigma,  \quad 
U_{\theta}= U^{(1)}_{\theta} \cup U^{(1)}_{\theta}\subseteq \M
$$

The cardinality of $\Delta_{\theta}$  is bounded by the integer $K$ which is 4 times the number of inequalities
in expression \ref{eq:negative-clauses} because each inequality of the form $x\not=1$
introduces one generator and  each inequality of the form $x\not=y$
introduces four not necessarily different generators.
Observe that  $\theta$ satisfies $\Phi$ if and only if  
so does $\pi\circ \theta$ where $\pi$ maps all generators in 
$\Sigma \setminus (\Delta\cup \Delta_{\theta})$ to the empty 
trace, so that we can, from now on, assume that $\theta$
maps each variable to the submonoid generated by 
$\Delta\cup \Delta_{\theta}$.

Let 
$ R_{\theta} $ 
be the predicate that specifies for all pairs of generators $(e,f)$
in $\Delta\cup \Delta_{\theta}$ whether they are equal or different
and which sub-alphabet $\Sigma_{i}$ they belong to.
A permutation $\sigma$ of $\Sigma$ is \emph{respectful}
if it fixes each element of $\Delta$
 and if it 
respects the membership to a specific $\Sigma_{i}$, i.e., if 1)
$\sigma(e)=e$ if $e\in \Delta$ 
and  2) for all $e\in \Sigma$ and for all $i=1, \ldots, n$ we have  $e\in \Sigma_{i}$ if and only if
$\sigma(e)\in \Sigma_{i}$. Now observe that up to a respectful permutation, 
the number of possible predicates $ R_{\theta}$ is finite: this is due to the fact
that $R_{\theta}$ involves a number of generators bounded by a function of the size of the formula
$\Phi$. Furthermore, if $\theta$ satisfies $\Phi$ so does $\sigma \circ \theta$. 
Consequently, when $\theta$ ranges over the possible assignments 
satisfying $\Phi$, up to a respectful permutation there exists a finite
number of different predicates  $R_{\theta}$, say  $R_{1}, \ldots, R_{\ell}$,
and each predicate involves a finite number  of elements of $\Sigma$.

Consequently, 
if $\Phi$ is satisfiable,  it is satisfiable by some assignment $\theta$
which maps the variables into the submonoid generated
by a finite subset $\Gamma_{i}$ containing $\Delta$ and satisfying a predicate $R_{i}$ 
for some 
$i\in \{1, \ldots, \ell\}$. We show that the existence of such an assignment
is decidable. Indeed, observe that the actual subset $\Gamma_{i}$
is irrelevant as long as it satisfies $R_{i}$. This means that we can consider 
the elements $a$ and $b$ as in expressions \ref{eq:x-not-1}, \ref{eq:x-not-y1}, 
\ref{eq:x-not-y2} and \ref{eq:x-not-y3} as fixed constants of $\Sigma$.
It remains to define the existential formulas  $\Phi_{i}$ 
in the structure $\langle \M(\Gamma_{i},I\cap (\Gamma_{i}\times \Gamma_{i}); \cdot, \+K_{i}\rangle$
as in the above claim \ref{eq:reduction-to-finitely-many-expressions}. 
This is achieved as follows. We keep the clauses \ref{eq:positive-clauses}, 
modify the clauses in \ref{eq:negative-clauses} as below and 
prefix the resulting formula by as many existential quantifiers as there are 
new variables. Concerning the modification of  
the clauses in \ref{eq:negative-clauses}, each inequality $x\not=1$ is replaced by a condition
$$
x=a z_{1}
$$
where $z_{1}$ is a new variable. 
Similarly,  each inequality  $x\not=y$ in \ref{eq:negative-clauses}
is replaced by a disjunction
$$
\begin{array}{l}
(y=x a z_{1})    \vee
(x=z_{1} a z_{2},  \wedge y =z_{1}  z_{3} b z_{4}  a z_{5} \wedge
z_{3}, z_{4} \in \M(\Gamma_{a},  I\cap (\Gamma_{a} \times \Gamma_{a}))) \\
\vee (x =z_{1} az_{2} \wedge y= z_{1}z_{3} \wedge z_{3} \in \M(\Gamma_{a},  I\cap (\Gamma_{a} \times \Gamma_{a}))) 
\end{array}
$$
where $\Gamma_{a}= \Gamma_{i} \setminus \{a\}$ and $z_{1}, \ldots,  z_{4}$
are new variables. Observe that no condition on the generators such as $a$ and $b$
above is required because these conditions are already covered by the 
predicate $R_{i}$.

We may now safely apply the result  \cite{diekert99} because  the only new predicates
$\M(\Gamma_{a},  I\cap (\Gamma_{a} \times \Gamma_{a})) $ are clearly regular 
in the finitely generated trace monoid $\M(\Gamma_i,  I\cap (\Gamma_i\times \Gamma_i)) $.

 \end{proof}

\section{An application to ordinals}

We denote by Ord the class  of ordinals. For a thorough exposition of 
ordinals we refer to the classical handbooks such as \cite{sie} and \cite{ros}.

\subsection{Arithmetic operations on the ordinals}

The following definition of the \emph{Cantor normal form}, abbreviated CNF,  is actually  a property
in its own right.

\begin{definition}
\label{de:cantor-normal-form}
Every nonzero ordinal $\alpha$  has a unique form as a sum of
$\omega$-powers with integer coefficients, namely
$$
\alpha = \omega^{\lambda_{r} } a_{r} +  \cdots + \omega^{\lambda_{1}} a_{1}, \quad 
$$
where $\lambda_{r} >  \cdots  >\lambda_{1}\geq 0$ are ordinals and
$a_{r},   \ldots, a_{1} >0$ are integers. 
A nonzero ordinal is a \emph{successor} if $\lambda_{1}=0$, otherwise it is a \emph{limit}. 

\end{definition}

\medskip We recall the definition of the multiplication on  ordinals
by use of their Cantor normal form.

%


%
\begin{definition}
\label{de:product}
Let 
$$
\alpha = \omega^{\lambda_{r} } a_{r} +  \cdots + \omega^{\lambda_{1}} a_{1}, \quad 
\beta = \omega^{\mu_{s} } b_{s} +  \cdots + \omega^{\mu_{1}} b_{1}, 
$$
be two nonzero ordinals written in CNF.
If  $\mu_{1}>0$ we have
%
\begin{equation}
\label{eq:product1}
\alpha \times \beta =  \omega^{\lambda_{r} + \mu_{s} } b_{s} +  \cdots + \omega^{\lambda_{r} +\mu_{1}} b_{1}.
\end{equation}
If  $\mu_{1}=0$ we have 
%
\label{eq:product2}
\begin{multline}
\alpha \times \beta  =  \omega^{\lambda_{r}+\mu_{s}} b_{s} + \omega^{\lambda_{r}+\mu_{s-1}} b_{s-1} + 
 \cdots + \omega^{\lambda_{r}+ \mu_{2}} b_{2} + \\
 \omega^{\lambda_{r} } a_{r}b_{1} +  \omega^{\lambda_{r-1} } a_{r-1} +\cdots + \omega^{\lambda_{1}} a_{1}. 
\end{multline}
\end{definition}

\begin{remark}
The  multiplication is associative, has a neutral element $1$, is noncommutative,  
is left- (but not right-) \emph{cancellative} ($x\times y = x\times z\Rightarrow y=z$)
and left- (but not right-) distributes over the addition. With the definition 
of the multiplication it can be easily verified that an ordinal $\lambda$ is 
closed under multiplication ($\alpha, \beta<\lambda$ implies $\alpha\times  \beta<\lambda$)
if and only if it is of the form $\omega^{\omega^{\xi}}$ for some ordinal $\xi\geq 0$.
Also since we are concerned with effectivity, we assume that $\lambda$ is less than 
the ordinal $\varepsilon_{0}$ so that providing a Cantor Normal Form 
and performing operations such as comparing ordinals and finding divisors make sense. 
\end{remark}

\subsection{Primes}

\begin{definition}
\label{de:primes}
An ordinal $x$ is a \emph{prime}
if it has exactly two right divisors, i.e., 
two ordinals $z_{1}\not=z_{2}$ for which there exist
$y_{1}, y_{2}$ with
$x=y_{1}z_{1}=y_{2}z_{2}$.
\end{definition}

This definition of prime is equivalent, as can be readily verified, to the standard 
definition which stipulates that it  has exactly two right divisors $1$ and $x$.

\begin{definition}
\label{de:three-kinds-of-primes}
There are three kinds of primes, \cite[p. 336]{sie}.

\begin{itemize}

\item finite primes: the ordinary prime natural numbers

\item non-finite successor primes: of the form $\omega^{\lambda}+1$, $\lambda\in \text{ Ord}$.

\item limit primes: of the form $\omega^{\omega^{\xi}}$, $\xi\in \text{ Ord}$
\end{itemize}
\end{definition}

The main result concerning primes is the following
\begin{theorem}[the prime factorization \cite{jacob}]
\label{th:jacob}%
Every ordinal
has a unique factorization of the form
$$
(\omega^{\omega^{\xi_{1}}})^{n_{1}}\cdots (\omega^{\omega^{\xi_{r}}})^{n_{r}}
a_{k}(\omega^{\lambda_{k-1}} +1)a_{k-1}(\omega^{\lambda_{k-1}} +1) \cdots 
a_{1}(\omega^{\lambda_{1}} +1)a_{0}
$$
with $\xi_{1}> \xi_{2} > \cdots > \xi_{r}$ and 
	$\lambda_{1}, \lambda_{2}, \ldots, \lambda_{k} \geq 1$ (Greek letters are arbitrary ordinals and Latin letters are finite ordinals).
\end{theorem}
 Observe that the condition on the exponents of limit primes is necessary: $\omega$ and
$\omega^{\omega}$ are primes and 
$\omega^{\omega}= \omega \times \omega^{\omega}$. 

\subsection{The monoid of successor ordinals as a trace monoid}

For every countable ordinal $\lambda$ closed under multiplication and less than $ \varepsilon_0$, let  $\+S_{\lambda}$ denote the set of all successor ordinals less than $\lambda$. 
The rule of multiplication on ordinals show that 
$\+S_{\lambda}$ forms a multiplicative submonoid.

Theorem \ref{th:jacob} can be interpreted as follows. The ordinal $\omega$
is the submonoid  generated by the finite  primes (the ordinary prime integers). Let $\+P_\lambda$
be the submonoid generated by the infinite successor primes less than $\lambda$. 
Then Theorem \ref{th:jacob}  claims that $\+S_\lambda$ is the free product 
of the  (infinitely generated) free commutative monoid $\omega$
and the (infinitely generated) free monoid $\+P_\lambda$.
We have 

\begin{theorem}
\label{th:exists-fragment-successors}
Given an ordinal $\lambda$ less than $\varepsilon_0$ and closed under multiplication, the existential theory of the structure $\langle \+S_{\lambda}:  \times,   \{\alpha\}_{\alpha \in \+S_{\lambda}}   \} \rangle$ is decidable.
\end{theorem}

\begin{proof}
Indeed, this is an immediate consequence of Theorem \ref{th:infinite-case} once we have 
observe the simple following relationship between the Cantor normal form and 
the prime factorization
$$
\omega^{\lambda_{r} } a_{r} +  \cdots + \omega^{\lambda_{1}} a_{1} + a_{0}
= a_{0}(\omega^{\lambda_{1}}+1)a_{1}(\omega^{\lambda_{2}-\lambda_{1}}+1)\cdots (\omega^{\lambda_{r}-\lambda_{r-1}}+1)a_{r}
$$
\end{proof}

\section{Open questions}

Via the unicity of the factorization for ordinals, the monoid of  successor ordinals
bears a strong resemblance to the free monoid: if one ignores the finite ordinals, 
one is left with a free infinitely generated monoid. This explains why, in the end, 
thanks to the improvements in \cite{schulz90} and \cite{diekert99}, we could 
resort to Makanin's result to obtain Theorem \ref{th:exists-fragment-successors}. 
What about the monoid of all ordinals? Or less ambitiously, what about 
solving equations in the monoid of ordinals with constants? 
More precisely, we are given  an equation 
$$
L(x_{1}, \ldots, x_{n}, \alpha_{1}, \ldots, \alpha_{p})= 
R(x_{1}, \ldots, x_{n}, \alpha_{1}, \ldots, \alpha_{p})
$$
where $L$ and $R$ are products of variables in 
$\{x_{1}, \ldots, x_{n}\}$ and constants  in $\{\alpha_{1}, \ldots, \alpha_{p}\}$
which belong to an ordinal $\lambda$ closed under multiplication.
We are looking for solutions where all variables take on 
nonzero values.

Without Makanin's result, the authors do not  know how to answer the question
even in  the specific case of the monoid of successor ordinals.
But this is maybe no indication that solving equations in 
the structure \str{\lambda;\times} with arbitrary constants
is at least as conceptually difficult as proving Makanin's result from scratch. 
We just make a couple of more or less trivial observations which tend to
show that the similarity of solving solutions in two structures (finite free monoids and
multiplicative ordinals) is maybe delusive.
For example, even if all constants are successor ordinals, it might be the case 
that the equation has no solution in the monoid of successor ordinals 
but has a solution  in $\lambda$, for example  $2x=x$.
More generally 
 it is not difficult, but boring, to prove that given an equation where the constants are limit ordinals,
it is  decidable in polynomial time relative to the number of unknowns
 whether or not it has a solution where the unknowns are themselves 
limit ordinals. Of course the equation could have only solutions in 
the successor ordinals  even if the constants were limit ordinals,
see $\omega x=\omega (\omega +1)$. More generally, considering specific 
submonoids $\Gamma$ of $\lambda$ for the constants such as the successors, the $\omega$-powers, 
the limit ordinals etc \ldots, and specific submonoids $\Xi$ for the values assumed 
by the variables, one can investigate whether or not an equation with constants 
in $\Gamma$ has a solution in $\Xi$.

\end{document}